\documentclass[conference]{IEEEtran}

\makeatletter

\def\ps@headings{%

\def\@oddhead{\mbox{}\scriptsize\rightmark \hfil \thepage}%

\def\@evenhead{\scriptsize\thepage \hfil \leftmark\mbox{}}%

\def\@oddfoot{}%

\def\@evenfoot{}}

\makeatother

\usepackage{graphicx}
\usepackage{url}
\usepackage{epstopdf}
\usepackage{algo}
\usepackage{amsgen}
\usepackage{amsfonts}
\usepackage{amsbsy}
\usepackage{amssymb}
\usepackage{amsmath}
\usepackage{amsthm}
\usepackage{cite}
\usepackage{subfigure}
\usepackage{epsfig}
\newcommand{\beq}{\begin{equation}}
\newcommand{\eeq}{\end{equation}}
\newcommand{\beqa}{\begin{eqnarray}}
\newcommand{\eeqa}{\end{eqnarray}}
\newcommand{\beqan}{\begin{eqnarray*}}
\newcommand{\eeqan}{\end{eqnarray*}}

\def\BibTeX{{\rm B\kern-.05em{\sc i\kern-.025em b}\kern-.08et\kern-.1667em\lower.7ex\hbox{E}\kern-.125emX}}
\newtheorem{thm}{Theorem}
\newtheorem{lemma}{Lemma}
\newtheorem{proposition}{Proposition}
\newtheorem{definition}{Definition}

\usepackage{multicol}
\begin{document}
\title{ Using Social Information for Flow Allocation in MANETs}
\author{
Andrew Clark$^1$, Amit Pande$^2$, Kannan Govindan$^2$,
Radha Poovendran$^1$ and Prasant Mohapatra$^2$\\
$^1$ Department of Electrical Engineering, Univ. Washington, USA \{awclark, rp3\}@u.washington.edu\\
$^2$ Department of Computer Science, Univ. California Davis, USA \{amit, gkannan, prasant\}@ cs.ucdavis.edu\\

}

\maketitle
\begin{abstract}
Ad hoc networks enable communication between distributed, mobile wireless nodes without any supporting infrastructure. In the absence of centralized control, such networks require node interaction, and are inherently based on cooperation between nodes. In this paper, we use social and behavioral trust of nodes to form a flow allocation optimization problem.
% address the problem of flow allocation in ad hoc networks based on each node's estimate of the trustworthiness of the other nodes. We first develop trust metrics and then formulate a flow allocation optimization problem incorporating the metrics.
We initialize trust using information gained from users' social relationships (from social networks) and update the trusts metric over time based on observed node behaviors. %The identity spoofing metric is calculated based on the trust between nodes, while flow diversity is incorporated to improve robustness to link failures and undetected malicious nodes.
We conduct analysis of social trust  using real data sets and used it as a parameter for performance evaluation of our frame work in ns-3. Based on our approach we obtain a significant improvement in both detection rate and packet delivery ratio using social trust information when compared to behavioral trust alone. Further, we observe that social trust is critical in the event of mobility and plays a crucial role in bootstrapping the computation of trust.
\end{abstract}

\section{Introduction}
\label{sec:intro}
Ad hoc networks  consist of wireless nodes, such as sensors, tablet computers, and smart phones, operating in the absence of supporting infrastructure.
 Due to the lack of a centralized authority, each node relies on its neighbors to identify multi-hop routes to its destination and forward packets along each route.
 At the same time, the intermediate nodes may exhibit selfish behavior by dropping packets in order to conserve scarce bandwidth and device resources.
	Furthermore, ad hoc networks may be deployed in the presence of malicious nodes, who join the network by masquerading as valid nodes and selectively drop, reroute, or tamper with packets~\cite{karlof2003secure}.

A common approach to mitigating selfish and malicious activity in ad hoc networks is through trust management~\cite{kannan_trust}.
In a trust management system, each node observes the behavior of each of its neighbors over a period of time and records instances of suspicious behavior, such as packet dropping, failure to follow MAC layer protocols, and broadcasting inaccurate routing information.  The nodes then form an empirical estimate of each neighbor's trustworthiness, denoted as a \emph{trust metric}.  By aggregating trust metric data from multiple nodes in a distributed fashion, each node can compute a global metric for the trustworthiness of each other network node, based on a combination of firsthand observation and reports from others.
%	In a trust management system, each node observes the behavior of each of neighboring node and estimates each node's reliability based on the observations [reference].
%	The nodes may also exchange trust data among themselves in order to rapidly identify malicious behavior.

Successful deployment of a trust management system carries several requirements.  First, the nodes must have a mechanism for determining whether packet drops and other suspicious events occur due to node misbehavior or communication and hardware failures.  Second, each node must have sufficient information to evaluate its neighbors' trustworthiness, in spite of mobility and network dynamics.  Third, the communication overhead of the network imposed by exchanging trust information must be minimized.  	
%Trust management systems have several disadvantages that have impeded the rapid deployment of ad hoc networks in defense and commercial applications.
%	First, it is difficult to determine whether packet drops and other suspicious events occur due to node misbehavior or communication and hardware failures.
%	Second, in order to form an accurate estimate of a node's trustworthiness, the node must be observed over an extended period, which may not be feasible in a mobile, dynamic network.
%	Third, the information exchanged in order to compute trust metrics increases the communication overhead of the network.

 In this paper, our insight is that the network users have existing social relationships with each other. When trust relationships between two nodes  exist, the nodes will not exhibit their selfish or malicious behavior. Such values can be quantified based on cached values of other contexts, in particular social networks. Social network data may be cached from online social networks, such as Facebook, LinkedIn, and Orkut, or may come from mobile ad hoc social networks, as proposed in~\cite{li09mobisn,sarigl09,gurecki09}.
 %in particular online social networks such as Facebook, LinkedIn, and Orkut.
	By leveraging the trust data accumulated through social relationships, a network node can more easily differentiate between benign failures due to device constraints and malicious behavior.
	This  reduces the amount of data collection and information exchange required to identify selfish or malicious nodes.
	Currently, however, there is no approach to incorporating social trust information in ad hoc network operation.

	We make two specific contributions:
	First, we develop methods for computing trust metrics based on social network data and composing social and behavioral trust metrics.
    We evaluate our trust metrics using experimental Facebook profile and wall post data.
	%Our methods lead to efficient computation of trust, as well as incorporation of behavioral trust data from other nodes.
		Second, we apply this trust information into a practical problem of flow allocation in MANETs. To this end, we propose a distributed optimization approach to allocating flows among multiple paths based on both network performance and the trust metrics of each path.
	
	We formulate flow allocation problems for two classes of utility, namely maximizing throughput and maximizing weighted flow diversity, and prove that our optimization algorithms lead to optimal flow allocation among trusted paths subject to capacity constraints.  Our approach is implemented through distributed network protocols and verified through ns-3 simulation study.  Our simulation results show that our approach provides high throughput and packet delivery even in the presence of malicious nodes. We obtain a significant improvement in both detection rate and packet delivery ratio using social trust information when compared to behavioral trust alone. Further, we observe that social trust is critical in the event of mobility (where behavioral values are not available) and plays a crucial role in bootstrapping the computation of trust.
	%We implement our approach through distributed protocols, and verify our approach through simulation study.  We show that our approach provides superior throughput and packet delivery in the presence of malicious nodes.

The paper is organized as follows.   In Section \ref{sec:model}, we state our system model and definitions of social network trust.    Our proposed trust-based optimization framework and flow allocation protocols are described in Section \ref{sec:flow_alloc}.  Simulation results are provided in Section \ref{sec:simulation}.  Related work is reviewed in Section \ref{sec:related}.  Section \ref{sec:conclusion} concludes the paper.

\begin{figure}[t]
	\centering
		\includegraphics[scale=0.35]{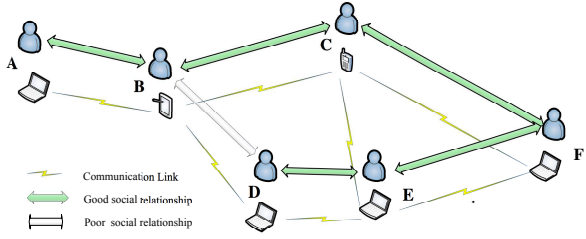}
	\caption{Simple diagram of an ad-hoc network. Six nodes are connected ad-hoc with each other via WiFi. Social relationships co-exist amongst some users.}
	\label{fig:manet_block_diagram}
\end{figure}

\section{Motivation}
\label{sec:motivation}

Wireless communications  and Internet have become commonplace and have completely transformed the social lives of people. Social networking sites such as Facebook (with over 800 million users), Google+, LinkedIn, Orkut (India, Brazil) and Renren (Chinese) have changed the paradigm of adolescent mating rituals, political activism, corporate management styles, classroom teaching and other dimensions of social life. In other words, social networking has made people more social inside the `net', than perhaps outside, in the real-world. This trend is on the increase.

More and more people are using these websites for social interactions and keeping others updated about themselves. Traditionally, people used these `connections' for communicating message from one person to another. People would leverage this social trust on one-another to communicate mails/letters or parcels to another person. In earlier times, if A is going to some city B, his friends will give him any parcels or mails/letters to their acquaintances in that place. Now, social relationships, being formed in online communities, we explore if these online social relationships can similarly help us in better `communication' amongst ourselves.
Basically, we explore a converse question: \textbf{\textit{Can social relationships affect the way wireless nodes interact with each other?}}

Consider the scenario given in Fig~\ref{fig:manet_block_diagram}. The active links between wireless devices are as shown in the figure. The social relationships between users are shown by green or red connections. It can be observed in this case that B has a good relationship with A and C, but poor relationship with D. Therefore,  D may drop a packet forwarded by B. Thus, poor social trust values are an indicative of `selfish' node behavior and higher social trust values will oblige a node to properly route the packets of another node. Such trust values may be derived from a direct relationship between the users, a relationship with a mutual acquaintance, membership in the same organization, or from online social network data. Unlike behavioral trust, social trust values exist independent of the ad hoc network, and are therefore not sensitive to changes in network topology or communication link failures. These social relationships in fact allow us to compute social trust values which empower ad hoc networks for rapid deployment.

In this paper, we consider a network where ad hoc nodes in ad hoc networks are member of social networks. Such combined networks are under active research consideration these days. CenseMe where mobile phones create mobile sensor networks and share the sensed information in social networks (Facebook) is one architecture in this direction. ~\cite{li09mobisn,sarigöl09,gurecki09} implement social networks for ad hoc nodes.

Use of trust relationships in this combined network addresses  two challenges. First, based on the behavioral trust on the user, we can determine (probabilistically), whether the packet drop is because of random failure (in case of high trust value), or due to malicious behavior. Secondly, although mobility may obstruct behavioral trust computation, social trust can be calculated based on existing social 	 relationships between users and ease trust assessment. 
\section{Trust Computation}
\label{sec:model}
As we mentioned earlier, trust is computed as a combination of behavioral and social trust values.
\begin{table}[t]
\caption{Parameters used in behavioral trust calculation}
\label{message_table_1}
\centering % used for centering table
 \begin{tabular}{|p{1.5cm}| p{3cm}| p{3cm}|}
    \hline
    \textbf{Field}&\textbf{Description} & \textbf{Value}\\
    \hline
    $\alpha$=$\alpha_1$ & evidence of good behavior count & Integer\\
    \hline
    $\beta$=$\beta_1$ & evidence of bad behavior count & Integer\\
    \hline
    $b$ & belief about the nodes action (trust) & $\alpha(1-u)/(\alpha+\beta)$\\
    \hline
    $d$ & disbelief about the nodes action (distrust=1-trust) & $\beta(1-u)/(\alpha+\beta)$\\
    \hline
     $u$ & uncertainty on the nodes trust calculations & $12\times Var(beta(\alpha, \beta))$\\
    \hline
 \end{tabular}
\end{table}

%{\bf Calculation of trust:}

\subsection{Calculation of behavioral trust}
\label{behave_trust_12}
Behavioral trust is measured by observing the behavior (actions) of one-hop neighbors. %Behavioral trust is also called as direct trust in literature and  is generally one hop in nature.
Behavioral trust is context-based and depends on the operational environment. The \emph{observational evidences} $\alpha_o$ (positive evidence) and $\beta_o$ (negative evidence) are collected based on the neighbors' actions (behaviors).

We illustrate the evidence based behavioral trust calculation by using the following example.
%\subsubsection{Illustrative Example-1: Routing based evidence collection}
%\label{routing_t_1}
Let us consider a packet routing scenario in a MANET and let node `$j$' is a target node (trustee node) whose trust has to be calculated. All neighboring nodes within one hop proximity of $j$ can overhear both the incoming packets and the outgoing packets at $j$. The routing table is pre-agreed upon and available at every node and hence a neighboring node of $j$ can easily identify the right destination for each packet. These neighboring nodes (trustor nodes) can check the destination of each outgoing packet from node $j$ and compare the destination against the routing table available with them. If the packet advances towards the right destination as per the routing table, then the routing behavior of the trustee node is considered to be correct and trustor can increment the $\alpha_o$ evidence on the trustee. If the packet does not get routed, gets corrupted or modified, or gets routed along an incorrect path, then the experience is recorded as a misbehavior $\beta_o$ evidence on the trustee node. Here we assume that the trustor node can observe both the incoming and outgoing packets at trustee node.

If all transmissions are kept constant at pre-decided power levels ($P_0$), the trustor nodes can reasonably estimate the received signal strength (RSS) and signal-to-noise (SNR) values at trustee node to back-calculate approximate channel and hence packet loss rate ($plr$). Genuine packet loss rate(PLR) due to channel losses can be approximated as:
\[PLR= plr \times (\alpha_0 + \beta_0)\]
We can adjust $\alpha$ and $\beta$ values to account for packet loss as follows:
\[\alpha_1=\alpha_0 + PLR \textrm{    ,  } \beta_1=\beta_0 - PLR\]

\subsubsection*{Determination of trust based on behavior}
The observed evidence values $\alpha_1$ and $\beta_1$ are used as parameters of a Beta distribution as follows:
\begin{eqnarray}
P(x)=B(\alpha_1, \beta_1, x)=\frac{\Gamma(\alpha_1+\beta_1)}{\Gamma(\alpha_1)\Gamma(\beta_1)}x^{\alpha_1-1}(1-x)^{\beta_1-1}
\label{beta_eqn}
\end{eqnarray}
where $\Gamma(.)$ is the Gamma function. Now the trust opinion will be generated as a triplet $<b, d, u>$, where $b$ stands for node's belief that the neighbor will behave properly in next time, $d$ is the node's disbelief of good behavior and $u$ is the uncertainty of the opinion. The expectation of the beta distribution ( $\alpha_1/(\alpha_1+\beta_1)$), is used to derive $b$ and $d$ and the variance ($Var(t)=\alpha_1\beta_1/(\alpha_1+\beta_1)^2(\alpha_1+\beta_1+1)$), normalized by factor $12$ is used to derive $u$ as shown in Table. \ref{message_table_1}. %For brevity we refer to \cite{trust_cal} for detailed calculation of $b$, $d$ and $u$.
This belief value $b$ is used as behavioral trust in our work.

We validate this trust model using empirical simulations. A scenario of $40$ nodes deployed over a square region of 100X100 m was assumed.  Ad hoc On-demand Multipath Distance Vector (AOMDV) routing algorithm is used for packet forwarding \cite{marina2002ad}. The behavior was averaged over 100 realizations. Our focus is on testing the effectiveness of our Trust Model, i.e., convergence and accuracy of the trust
calculation in determining the exact value of trust. From Fig. \ref{behavioral_routing_r1} we can see that on an average our approach converges to the correct trust value of these nodes in 14 event observations. The initial trust is always less due to uncertainty factor ($u$) which is high in the initial stage. However, as the number of evidences increases the uncertainty slowly reduces and hence the trust value converges to its original value.

%\begin{figure}
%	\centering
%		\includegraphics[width=3in]{figures/behavioural_trust_1.eps}
%\caption{Packet routing based behavioral trust calculation}
%\label{behavioral_routing_r1}
%\end{figure}

\subsection{Calculation of social trust}
As we have seen before, the convergence of behavioral trust takes certain time. If we need to make the trust decision quickly or if we cannot observe the one hop behaviors, we can use social media and derive trust based on social profiles of users.

Here, we are assuming that modes in adhoc network have access to a social network. This may not always be true and we may have to revert to behavioral trust only formulation in those scenarios. However, many mobile social networks such as Adsocial~\cite{sarigl09} or MobiSN~\cite{li09mobisn} are emerging lately. AdSocial targets small-scale scenarios such as  co-workers sharing calendar information. In one scenario, 30+ users used AdSocial to call each other from room to room, chat, and keep track of users’ location in the hotel~\cite{sarigl09}. More often than that, people are inter-connected using social networks such as Facebook or Google+. The offline, cached information from such networks can be used to obtain social trust value. In this paper, we demonstrate this with some scenario specific to Facebook users, although the analysis is generic.

\begin{figure*}[t]
\hspace{-0.15in}
\begin{minipage}[b]{0.33\linewidth}
	\centering
		\includegraphics[width=2.6in, height =1.6 in]{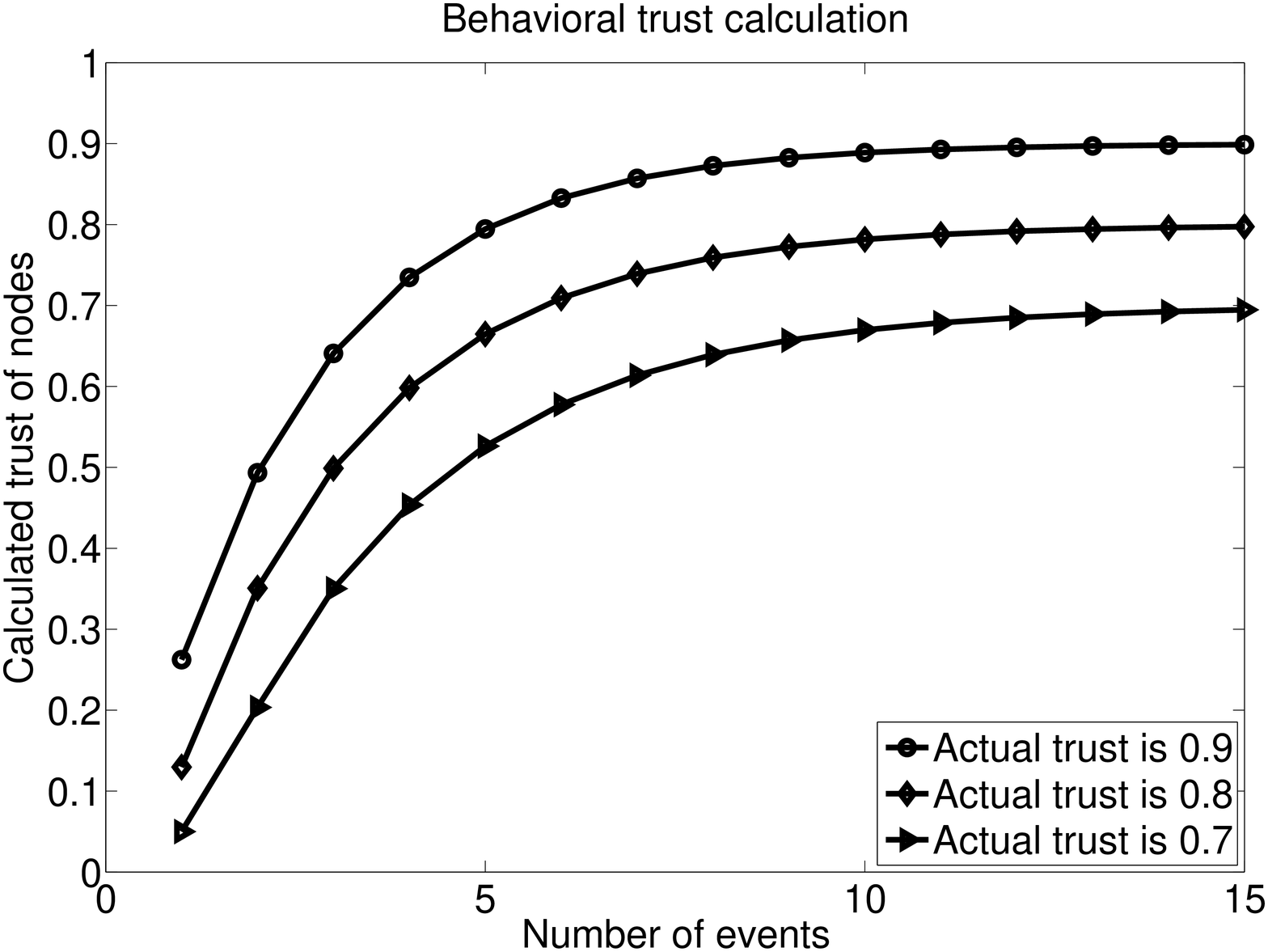}
\caption{Packet routing based behavioral trust calculation for three actual trust values (results based on ns-3 simulation)}
\label{behavioral_routing_r1}
\end{minipage}
\hspace{0.5cm}
\begin{minipage}[b]{0.33\linewidth}
	\centering
		\includegraphics[width=2.6in, height =1.6 in]{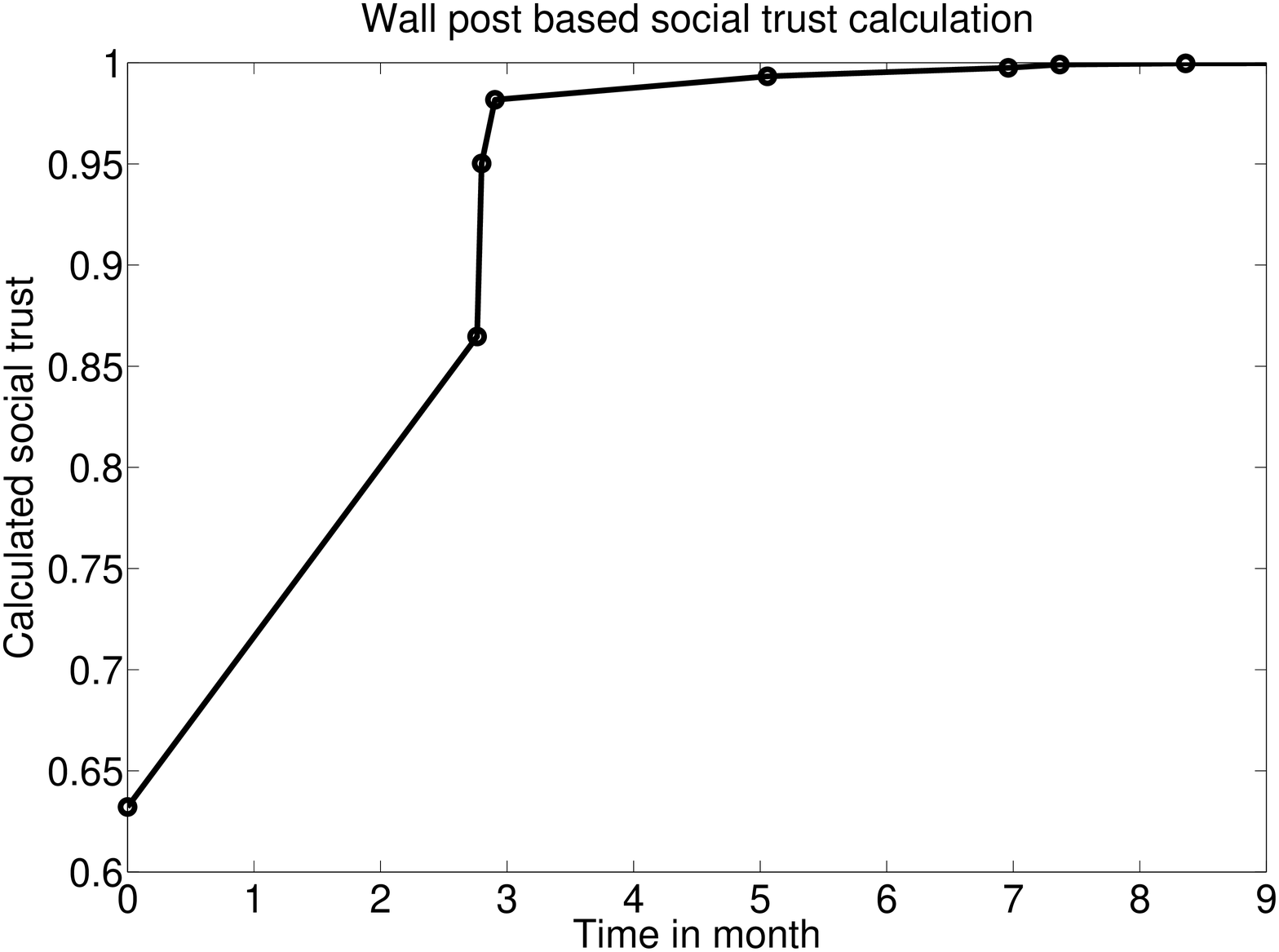}
\caption{Sample social trust computation using Wallpost interactions of Facebook users~\cite{Viswanath:2009}. Trust values stablize over a period of few months.}
\label{social_r1}
	\end{minipage}
\begin{minipage}[b]{0.33\linewidth}

	\centering
		\includegraphics[width=2.6in, height =1.6 in]{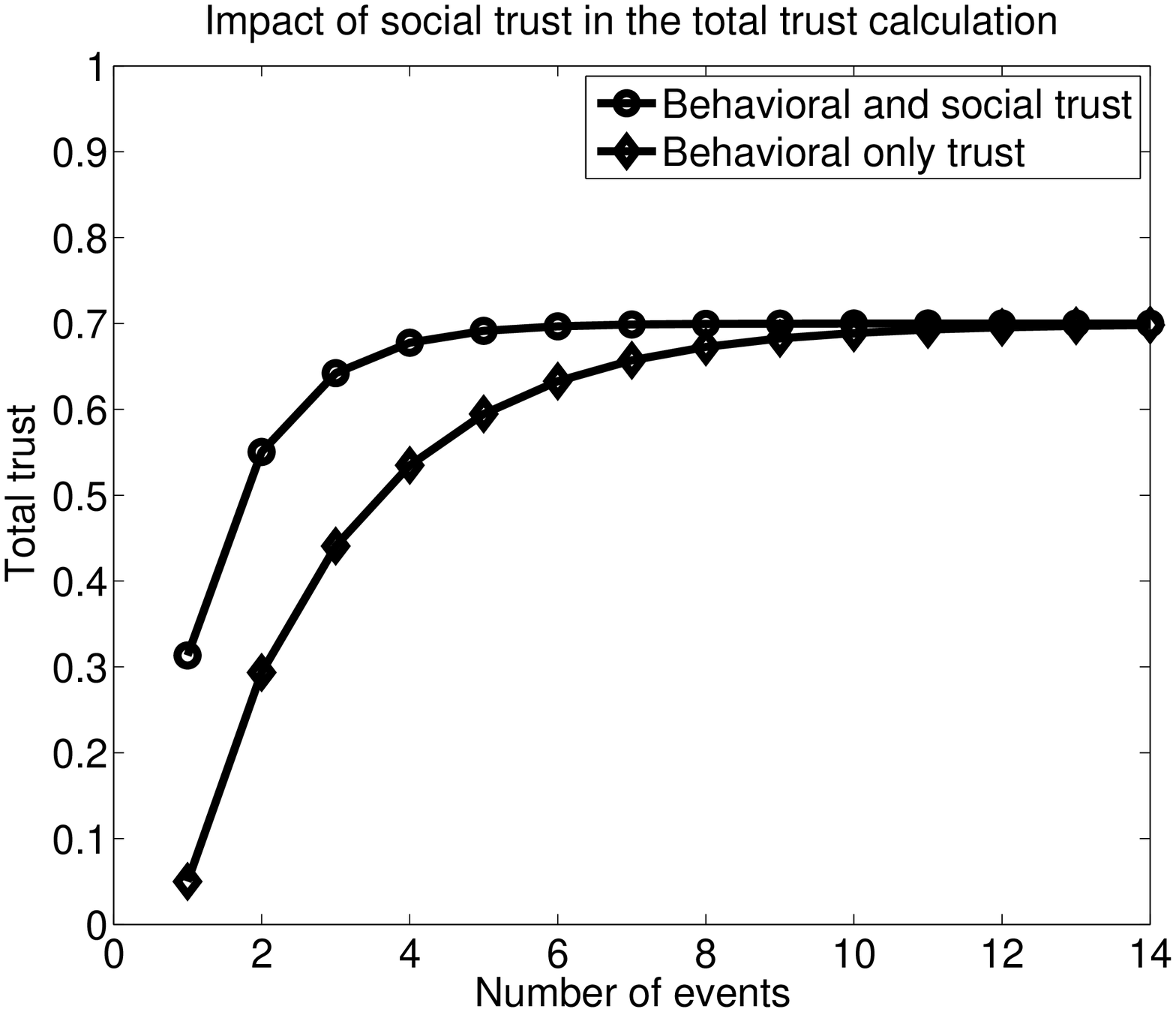}
\caption{Improvement in trust computation using social information (from Facebook data~\cite{Viswanath:2009}). Only 4 events are sufficient to estimate trust value (as against 10-14)}
\label{fig:social_behavioural_comparison}
	\end{minipage}
\end{figure*}

%We assume that the nodes are socially connected in an different domain than the purpose for which they have been employed. For example, people in a professional field may also be member of online social networks such as Face Book and Orkut, or ad-hoc social networks such as Adsocial~\cite{sarigöl09} or MobiSN and if we have an access to their social domain, we can effectively derive a trust factor based on social connectivity.

A typical user profile in an on-line social network is characterized by its profile {\it features} like {\it location, hometown, activities, interests, favorite music, professional associations}, etc. In sites like Facebook and Orkut, users establish connectivity and then friendships when they discover similar profile entries. In LinkedIn people connect amongst each other to build professional networks if they find profiles match in terms of affiliations, qualifications and work. It implies that people tend to make social connection with certain trust level if they find a person of similar profile in terms of {\it features}. Similarly, people interact more only when they find that the person is socially reliable to an extent. That means, the frequent social interaction can be used to derive trust.

There can be various measures of social trust amongst two users based on information obtained from online social networks. We demonstrate the concept using inter-profile similarity and wall post interactions of users. In a nut shell, if a user meets a strange person and receives information from a stranger, as long as the person is similar in terms of either location, interests, and other related {\it features}, then he tends to trust that person more. Also, a person continues to interact and exchange data with a person only with a trusted person.
%Some of the parameters which can be leveraged to obtain trust from the online social media are:
%
%\begin{enumerate}
%\item Frequent interaction between users in the online social media (for example wall post interactions in Facebook) and hence the connectivity between users.
%\item Similarity across different user profiles and correlation of user similarity with the network topology in on-line social networks.
%\end{enumerate}
With this hypothesis we follow the below steps in obtaining the social trust

%\begin{figure}
%	\centering
%		\includegraphics[width=3in]{figures/social_trust_1.eps}
%\caption{Sample social trust computation}
%\label{social_r1}
%\end{figure}

\begin{itemize}
\item If the two users are not connected socially (e.g., not friends in Facebook) than their inter profile similarity factor (IPS) will provide an estimate of social trust.
\item If the users are already socially connected (for example if they are friends in Facebook) then their relative frequency of wall post interactions will reflect the trust level. More frequent wall posting with a particular user compared to all other users, will translate to higher trust and fewer postings will translate to lower trust.
\item If the users are social connected and also frequently interacts then the combinations of their IPS and interaction factor will provide the social trust value.
\end{itemize}

{\bf Inter Profile Similarity (IPS):}
Inter profile similarity is defined as how similar two users are in terms of various semantics (metrics). The IPS is measured in the interval of (0, 1) where 1 denotes complete similarity and 0 denotes complete dissimilarity. Our IPS is measured using natural language processing (NLP) inspired by reference \cite{felix_ips}. IPS compares short text phrases between two profiles. The benefit of using NLP is that the different words used in the profile which has same meaning will be identified correctly. The learning process of the NLP also keep improving with the available data size.

In our work we use the social network data such as Facebook and Orkut as an illustrative examples for measuring the IPS factor. We assume that the profiles in social networks are described using a set of key words and we have access to those profile details. Facebook allows users to describe themselves in a number of different categories; however, we concentrate only on the following {\bf \it features}: (1) activities, (2) interests, (3) gender, and (4) affiliations. Using the affiliations, users are able to restrict the set of people that can view their profile, the default policy is that only those in the same affiliation or are immediate friends may view each other's profiles. %Similarly, in Orkut we concentrate on the following {\bf \it features}: (1) activities, (2) passions, (3) sex, and (4) communities.
We chose to focus on these categories as they generally use terms that exist in a dictionary.
%\begin{figure}
%	\centering
%		\includegraphics[width=3in]{figures/Social_Behavioural_Trust_1.eps}
%\caption{Impact of social trust on trust calculation.}
%\label{fig:social_behavioural_comparison}
%\end{figure}
These semantics have been analyzed and an IPS score of 0 to 1 calculated based on the algorithm proposed in \cite{felix_ips} for each features. If the IPS score is above $0.5$ then we increase the positive evidence count ($\alpha_s$) or else we increment the negative evidence count ($\beta_s$). Now using these $\alpha_s$ and $\beta_s$ we can determine the social trust as follows:

$I_{i,j}(t)$ is trust of $i$ on $j$ computed at time $t$ based on Inter-profile similarity. It is computed using a Beta distribution similar to behavioral trust.

\begin{equation}
I_{i,j}(t) = \frac{\alpha_s}{(\alpha_s+\beta_s)}(1-u)\textrm{, where }  u = 12 \times {\rm Var} ( B(\alpha_s, \beta_s) )
\label{social_ips_1}
\end{equation}
Similarity in user affiliations form the most important likelihood of user support in MANETs in absence of direct relationship  amongst users. In case, such relationship exists, we infer social trust based on wall post interactions.

{\bf Trust based on wall post interactions:}
In this section we derive the social trust based on the interactions in social networks. The social trust is derived based on the frequency of data exchanges between users. We use the hypothesis that activities (frequent data exchanges between users) generally represent a strong relationship in the social network and hence a strong trust. For illustration purposes we use wall post exchange in Facebook as the method of interactions. We analyze the distribution of the number of wall posts per link and investigate to what extent a user $i$ exchanges wall posts with a particular user $j$ compared to all other users. Let us assume $W_{i,j}(t)$ is trust of $i$ on $j$ at time $t$ based on wall-post interactions. It is computed as follows:
\begin{equation}
W_{i,j}(t) = 1- e^{-ax(t)}
\label{wall_post_trust}
\end{equation}
where,
\[x(t) =  N_{i,j}(t) \frac{1}{    \big(\frac{N_i(t)}{C} \big)  } \]

where $N_{i,j}(t)$ is the number of accumulated wall posts by $i$ on $j$'s wall up to time $t$, $N_i(t)$ is the total number of wall posts by $i$ up to time $t$ and $C$ is the total number of contacts (friends) of node $i$. Therefore, $\frac{N_i(t)}{C}$ gives the average wall post rate of $i$.

To analyze our proposed interactions based trust, we use the data set given by \cite{Viswanath:2009, social_data_1}. This is a wall post data between September 26, 2006 and January 22nd, 2009 collected in the Facebook network of New Orleans region. Each wall post entry in the data set contains information about the wall owner, the user who made the post, the time at which the post was made (based on Unix time stamp), and the post content. The data is wall postings of $60,290$ ($66.7\%$) users and $838,092$ wall posts, in New Orleans networks.

 Calculated trust based on wall postings against the time in months is shown in Fig. \ref{social_r1}. The number of wall posts varies randomly against time. However, from  Fig. \ref{social_r1} we observe that, based on our approach we can able to make the trust decision within $3$ months of interactions. That means a person has enough interactions within $3$ months of time to achieve full trust (trust=1).

The net social trust on $j$ by $i$ at time $t$ ($S_{i,j}(t)$) is computed as weighted average of the IPS and wall-Post trust values:

\begin{equation}
S_{i,j}(t) = \eta(t)\times W_{i,j}(t) + (1-\eta(t))\times I_{i,j}(t)
\end{equation}

where $0\leq \eta(t) \leq 1$ is a time dependent proportionality constant. As we can see from the wall post data results, on an average it takes few months of interactions to place a full trust on a person. Therefore, in the initial stage, it is better to give more weight to the IPS based trust than a wall post interaction based trust. Hence, $\eta(t)$ should assume a less value in the initial time and should increase as a function of both time and the activity factor (total number of wall posts the person shares with all other users).

\subsection{Computation of Overall Trust}
\label{subsection:joint_trust}
The overall trust metric is computed as a function of both social trust and observed behavior.  %We classify observed behavior into \emph{positive} and \emph{negative} events.  As an example, a node can observe its neighbors routing behavior.  If a packet is forwarded correctly within a bounded time interval, with the correct header information, then the event is recorded as positive.  Otherwise, if packets are dropped or routed to the wrong destination, the event is recorded as negative.
%Node $i$'s belief of node $j$'s trustworthiness is quantified through the parameter $T_{ij}$, which represents the probability that node $j$ can be trusted.
%Each node $i$ models the actions of its neighbors as an i.i.d. Bernoulli process, in which the neighboring node $j$ behaves positively with probability $T_{ij}$ and negatively with probability $T_{ij}$.
%The probability $T_{ij}$, which represents node $i$'s trust for node $j$,
The trust of node $i$ for node $j$ is represented by $T_{ij} \in [0,1]$, which represents the probability that node $j$ is trustworthy based on information available to node $i$.
 We model $T_{ij}$ as an unknown system parameter that is estimated based on the observed events, denoted $O_{1}, \ldots, O_{n}$, using log likelihood analysis.  The prior distribution of $T_{ij}$ is derived from social trust.  We have
\begin{IEEEeqnarray}{rCl}
\label{eq:trust_def}
\nonumber
T_{ij} &=& \arg\max_{x \in [0,1]}{F_{ij}(x | O_{1}, \ldots, O_{n})} \\
&=& \arg\max_{x \in [0,1]}{F_{ij}(O_{1}, \ldots, O_{m}|x)F_{ij}(x)}
\end{IEEEeqnarray}
where $F_{ij}$ is the probability density function of $T_{ij}$, derived using the social trust values. We take the IPS metric as an example, leading to a Beta distribution for $F_{ij}$ defined by $F_{ij}(x) = x^{\alpha-1}(1-T)^{\beta-1}/B(\alpha, \beta)$.  Under this model, $T_{ij}$ is given by the following proposition.

\begin{proposition}
\label{prop:trust_metric}
Assuming a Beta distribution for $F_{ij}(T)$, the trust metric computed based on observed events $O_{1}, \ldots, O_{n}$ is equal to
\begin{equation}
T_{ij} = \frac{r + \alpha -1}{n + \alpha + \beta-2}
\end{equation}
where $r$ is the number of positive events.
\end{proposition}

\begin{proof}
From (\ref{eq:trust_def}), $T_{ij}$ is chosen to maximize
\begin{displaymath}
F_{ij}(O_{1}, \ldots, O_{n}|x)F_{ij}(x) = x^{r}(1-x)^{n-r}x^{\alpha-1}(1-x)^{\beta-1}
\end{displaymath}
which is equivalent to maximizing 
\begin{displaymath}
(r+\alpha-1)\log{x} + (n-r+\beta-1)\log{(1-x)}
\end{displaymath}
since the log function is monotonic.  Setting the derivative equal to zero implies that the maximum is achieved when $T_{ij} = \frac{r+\alpha-1}{n+\alpha+\beta-2}$.
\end{proof}

Intuitively, Proposition \ref{prop:trust_metric} states that node $i$'s trust for node $j$ is the  number of positive observations (from both social and behavioral trust) normalized by the total number of social and behavioral observations. Computation of $T_{ij}(n+1)$, the observed trust after $n+1$ observations, can be performed using a simple linear update rule, where $n^{\prime} = \alpha + \beta -2$ can be stored as a static system parameter:
\begin{equation}
\label{eq:update_rule}
T_{ij}(n+1) = \left\{
\begin{array}{cc}
T_{ij}(n)\frac{n+n^{\prime}}{n+n^{\prime}+1}, & \mbox{if $O_{n+1}$ is negative} \\
\frac{T_{ij}(n)(n+n^{\prime}) + 1}{n+n^{\prime}+1}, & \mbox{if $O_{n+1}$ is positive}
\end{array}
\right.
\end{equation}

Furthermore, note that if there are no social trust data available, then the update rule (\ref{eq:update_rule}) is equal to the fraction of positive observations, as above.
%\subsection{Analysis of Trust Computation}
% \label{subsec:trust_analysis}
 %We analyze the joint trust computation through two performance metrics.  First, letting $T_{j}^{\ast}$ denote the true probability that node $j$ will behave negatively, we study $\mathbf{E}((\hat{T}_{ij}(n) - T^{\ast})^{2})$, the mean square error under our approach.  Second, we consider the probability of classifying nodes incorrectly, defined below.

 %\subsubsection{Mean-square error analysis}
 Letting $T_{j}^{\ast}$ denote the true probability of negative behavior by node $j$, the mean-square error of the trust computation is given by $\mathbf{E}((T_{ij}(n)-T_{j}^{\ast})^{2})$.  The mean-square error is described by the following lemma.

 \begin{lemma}
 Suppose that each observation of node $j$ is positive with probability $T_{j}^{\ast}$ and negative otherwise, and suppose that, based on social trust data, the prior distribution of node $j$'s behavior is a Beta distribution with parameters $\alpha$ and $\beta$.  Define $n^{\prime} = \alpha + \beta -2$ and $T_{j}^{\prime} = \frac{\alpha-1}{\alpha + \beta -2}$.  Then
 \begin{equation}
 \mathbf{E}((T_{ij}(n) - T_{j}^{\ast})^{2}) = \frac{nT_{j}^{\ast}(1-T_{j}^{\ast}) + n^{\prime 2}(T_{j}^{\ast} - T_{j}^{\prime})}{(n + n^{\prime})^{2}}
 \end{equation}
 \end{lemma}

 \begin{proof}
 Let $r$ be a random variable denoting the number of positive observations.  We have
 \begin{IEEEeqnarray}{rCl}
 \mathbf{E}((T_{ij}(n)-T_{j}^{\ast})^{2}) &=& \mathbf{E}\left(\left(\frac{r+\alpha-1}{n+n^{\prime}} - p\right)^{2}\right) \\
  &=& \mathbf{E}\left(\left(\frac{r+T_{j}^{\prime}n^{\prime}}{n+n^{\prime}} - p\right)^{2}\right)
  \end{IEEEeqnarray}
  Simplifying further, using the fact that $r$ is binomial with $\mathbf{E}(r) = nT_{j}^{\ast}$ and $\mathbf{E}(r^{2}) = nT_{j}^{\ast}(1-T_{j}^{\ast}) + n^{2}T_{j}^{\ast 2}$, yields the desired result.
 \end{proof} 
 
 The trust formulation doesn't involve any specific overhead in network transmissions because social trust values are computed using offline social networks data (although they can be updated if needed) while behavioral trust on peer node is calculated using passive monitoring.

\section{Trust-based Flow Allocation}
\label{sec:flow_alloc}
We demonstrate the use/ advantage of social trust values using the example of flow allocation in MANETs.
%Next, we incorporate the computed trust values to resolve the issue of flow allocation in mobile adhoc networks. 
We present a trust-based flow allocation optimization and  a distributed solution approach.
\subsection{System Model}
\label{subsec:comm_model}
We assume a set of users that communicate over a wireless network.  A wireless link $(i,j)$ exists between users $i$ and $j$ if $j$ is
within $i$'s wireless range; let $L$ denote the set of links.  Each link $l \sim (i,j)$ has a nonnegative capacity $c_{l}$.

Two (not necessarily disjoint) subsets $S,D \subseteq V$ of users act as data sources and destinations, respectively.
Each source $s \in S$ sends unicast traffic to a destination, denoted $d_{s} \in D$, with rate $r_{s}$.
When $d_{s}$ is not an immediate neighbor of $s$, $s$ maintains a set of paths, denoted $\mathcal{P}_{s}$, using an ad hoc
routing protocol \cite{jacquet2001optimized}.  Each path $\pi \in \mathcal{P}_{s}$ consists of a set of intermediate
links, so that $\pi = \{(s,i_{1}), \ldots, (i_{k}, d_{s})\}$.  Letting $r_{s,\pi}$ denote the rate of traffic from source
$s$ through path $\pi$, flow conservation and capacity requirements imply the following constraints:
\begin{IEEEeqnarray}{rCl}
\sum_{\pi \in l, s \in \mathcal{S}}{r_{s, \pi}} &\leq& c_{l}  \qquad \forall l \in L\\
\sum_{\pi \in \mathcal{P}_{s}}{r_{s, \pi}} &=& r_{s} \qquad \forall s \in S
\end{IEEEeqnarray}

where $\{\pi \in l\}$ denotes the set of paths that include $l \in L$ as an intermediate link.

%\subsection{Social network and user model}
%\label{subsec:social_model}
%A link $(i,j)$ in the social network exists if user $i$ has a trust relationship with user $j$;
%let $E^{\prime}$ denote the set of social network
%links.  Let $t_{ij} \in [0,1]$ denote the trust value of $i$ for $j$, representing $i$'s perceived probability
%that $j$ is trustworthy.

\subsection{Adversary Model}
We assume two types of users in the network, \emph{benign} and \emph{malicious}.
Benign users who are not source nodes attempt to maximize the global utility by forwarding all packets they receive. Benign users who are also source nodes forward all packets they receive, and attempt to maximize their source rate using any residual capacity.  We assume that benign users have access to cached social network data that is used to compute social trust metrics.
%Benign users who are also source nodes
%attempt to maximize their source rates, while benign non-source nodes attempt to maximize the global utility.
Malicious users,
on the other hand, attempt to reduce the throughput of one or more sources.  This can be done when the malicious users lie
on a source-destination path, and can drop or re-order packets.

%We explain these metrics in the following sections.

\subsection{Verifying user identities}
\label{subsec:asymmetric}
%The users of the network may employ cryptography in order to authenticate each other and check the integrity of certain packets.
The use of social trust data by node i to determine the trustworthiness of node j can be thwarted if a selfish or malicious
user masquerades as a user trusted by i. This can be incorporated into trust metrics by asking a set of nodes, denoted $R(j)$,
 to vouch for j's true identity, as described by the following metric.

%User $i$'s level of trust in $PK_{j}$, denoted $SM(i,j)$ represents the belief that $PK_{j}$ was actually created by user $j$ instead of a malicious user spoofing the identity of $j$.  It is defined as a function of the trust values of the nodes in $R_{j}$, as follows:
\begin{definition}
\label{def:asymmetric_metric}
The \emph{identity spoofing metric} $ISM(i,j)$ for link $(i,j)$ is defined to be the probability that
at least one of the users in $R_{j}$ is valid, given by
\begin{equation}
\label{eq:asymmetric_metric}
ISM(i,j) = 1 - \prod_{r \in R_{j}}{(1 - T_{i,r}ISM(i,r))}
\end{equation}
\end{definition}
where $T_{i,r}$ is the combined trust factor of $i$ on $r$ (combination of social and behavioral trust), and $ISM(i,r)$ is the trust of $i$ on $r$'s identity. It is assumed that $i$ has independently verified the identity of at least one user, $j$, resulting in $ISM(i,j) = 1$. This independence assumption is required to maintain a basic level of trust on each other's identity which is then propagated using combined trust metric ($T_{i,r}$). The values of $ISM(i,j^{\prime})$, for all $j^{\prime} \in V$, can then be computed using the fully-trusted users as a starting point.
 In (\ref{eq:asymmetric_metric}), $T_{i,r}ISM(i,r)$ is the probability that a node $r$'s claim of node $j$'s identity is trustworthy, equal to the probability that $r$ is a trustworthy user ($T_{i,r}$) times the probability that $r$'s identity is correct ($ISM(i,r)$). We observe that, if one of the nodes in $R_{j}$ is fully trusted ($T_{i,r} = ISM(i,r) = 1$), then $ISM(i,j)=1$ in (\ref{eq:asymmetric_metric}).

\subsection{Computing path trust}
\label{subsec:path_trust}
 A path is trusted if and only if each link in the path is trusted.  Assuming that the trustworthiness of each node $j$ in the path is independent of the trustworthiness of the other nodes in the path, the probability that the path on the whole is trustworthy is therefore equal to the product of the metric values for each node in the path.
\begin{definition}
\label{def:path_key_vulnerability}
Let $\pi = (i_{0}, \ldots, i_{n})$, where $i_{0} = s$ and $i_{n} = d$, denote a path between a source $s$ and a destination $d$.  The  path trust metric $T(\pi)$
is  equal to the product of the probabilities that each intermediate node in the path is trustworthy,
\begin{equation}
\label{eq:path_trust}
T_{\pi} = \prod_{k=0}^{n-1}{T_{i_{k}i_{k+1}}}
\end{equation}

%$SM^{s,d}(\pi) \triangleq \prod_{k=0}^{n-1}{SM(i_{k},i_{k+1})}$.
\end{definition}
Each source can compute the trust in the path by piggybacking trust values on routing control packets.

\subsection{Problem formulation}
\label{subsec:centralized}
Two utility functions are considered.  In the first, the goal of each source destination pair $(s,d)$ is to maximize the available throughput. This can be expressed by choosing the utility function $U_{sd}^{(1)}(\mathbf{r}_{s})$, defined as
\begin{equation}
\label{eq:utility}
U_{sd}^{(1)}(\mathbf{r}_{s}) = \sum_{\pi \in P_{s}}{\log{(1 + r_{s,\pi})}} %- \mu_{s}\frac{r_{s,\pi}}{r_{s}}\log{\frac{r_{s,\pi}}{r_{s}}}\right)}
\end{equation}
%Here ${\cal P}_s$ is the set of all paths from source s to its destination. We do not assume that the paths are node or link disjoint. $r_{s\pi}$ is the flow that is routed through path $\pi$ so that $r_{s} = \sum_{\pi \in  {\cal P}_s} r_{s\pi}$.
%where $P_{s}$ is the dimension of vector ${\bf  r}_s$ and parameter $\mu_s$  $(\mu_s \geq 0)$ controls the tradeoff between maximizing flow (the first term) and balancing flow (the second term).

In the second, the goal is to divide the flow among multiple paths, giving extra weight to the paths with higher trust value.  This goal of flow diversity is captured by the utility function
\begin{equation}
\label{eq:utility1}
U_{sd}^{(2)}(\mathbf{r}_{s}) = \sum_{\pi \in P_{s}}{-\frac{r_{s,\pi}}{T_{\pi}}\log{\frac{r_{s\pi}}{T_{\pi}}}}
\end{equation}

The utility function can also be a combination incorporating both throughput and diversity, with $U_{sd}(\mathbf{r}_{s}) = U_{sd}^{(1)}(\mathbf{r}_{s}) + \mu_{s}U_{sd}^{(2)}(\mathbf{r}_{s})$, where $\mu_{s}$ is a nonnegative constant that can be tuned to change the relative importance of each term.  Next, we state the optimization problem for multiple source-destination pairs. First, the flow through each link $l$ cannot exceed the link capacity constraint $c_{l}$.  Second, each path must meet a trust threshold $\tau_{T}$.  Third, each path must meet a  threshold on the probability that no identities have been spoofed, denoted $\tau_{S}$.  This results in the following optimization problem
\begin{equation}
\label{eq:flow_opt}
\begin{array}{l}
\max \>\>  \sum_{s \in S}{U_{sd}(\mathbf{r}_{s})}\\
r_{s\pi} \\
\\
\mbox{s.t.} \>\> \hspace{25pt} \sum_{s \in \mathcal{S}}{W_{s}\mathbf{r}_{s}} \leq \mathbf{c} \\
r_{s,\pi} = 0 \mbox{ if $\prod_{(i,j) \in \pi}{T_{i,j}} < \tau_{T}$}, \quad \forall s \in \mathcal{S}, \pi \in \mathcal{•}_{s}\\
% \hspace{25pt} r_{s,\pi} = 0 \mbox{ if $SM^{s,d}(\pi) < \tau_{S}$ } \\
%\prod_{j \in \pi}{f_{PK}(s,j)} \geq \tau_{V} \forall \pi \in P_{s}, s \in S \\
% & \prod_{(i,j) \in \pi}{t_{ij}} \geq \tau_{T} \forall \pi \in P_{s}, s \in S \\
 %& \sum_{s \in {\cal S}} \sum_{\pi_s \in {\cal P}_s } \sum_{\ell \in \pi_s} r_{s, \pi_s} \le c_\ell \qquad \forall \ell \in L

 \end{array}
 \end{equation}

where $W_{s}$ is the routing matrix for source $s$, $(i,j)$-entry is $1$ if link $i$ is traversed by path $j$ and $0$ otherwise, while $\mathbf{c}$ is the vector of link capacities.

 Eq. (\ref{eq:flow_opt}) is a concave optimization problem, and can therefore be solved efficiently by a centralized authority.  In practice, however, no such centralized authority exists, and $T_{i,j}$ changes over time.  We instead propose a distributed approach as explained below.

% \subsection{Distributed algorithms}

%At the network level, the main goal is to maximize throughput without violating capacity constraints and while maintaining the required levels of trust and security.  The expected overall throughput of a flow allocation, which determines the utility, is given by
%\begin{equation}
%U(\mathbf{r}) = \sum_{s \in S}{w_{s}\sum_{\pi \in P_{s}}{r_{s,\pi}}}
%\end{equation}
%where $w_{s} \geq 0$ are nonnegative weights representing the relative importance of each source.
%Combining this utility function with the  capacity and trust constraints yields the centralized optimization problem
%\begin{equation}
%\label{eq:flow_opt}
%\begin{array}{cc}
%\mbox{maximize} & \sum_{s \in S}{w_{s} \sum_{\pi \in P_{s}}{T_{\pi}r_{s,\pi}}} \\
%r_{s,\pi} & \\
%\mbox{s.t.} & \sum_{\pi \in l}{r_{s,\pi}} \leq c_{l} \qquad \forall l \in L \\
% & T_{\pi} \geq \tau_{T} \forall \pi \in P_{s}, s \in S
%\end{array}
%\end{equation}
%where $c_{l}$ is the capacity of link $l$ and $L$ denotes the set of links.  This optimization problem is a linear program, which can be efficiently solved.

\subsection{Distributed algorithm}
The utility function in (\ref{eq:flow_opt}) is the sum of the utility functions of each source $s$, each of which is a concave function of the source rates.  Distributed algorithms can therefore be found using dual decomposition methods \cite{chiang2007layering}.  The following analysis is based on the utility function $U_{sd}$, but also holds for $U_{sd}^{(1)}$ and $U_{sd}^{(2)}$.

The Lagrangian of (\ref{eq:flow_opt}) is given by
\begin{equation}
\label{eq:Lagrangian}
L(\mathbf{r}, \mathbf{\lambda}) = \sum_{s}{U_{sd}(\mathbf{r}_{s})} - \mathbf{\lambda}^{T}\left(\sum_{s}{W_{s}r_{s}} - \mathbf{c}\right)
\end{equation}

Let  $g(\mathbf{\lambda}) = \max_{\mathbf{r}}{L(\mathbf{r},\mathbf{\lambda})}$.  $g(\mathbf{\lambda})$ can be decomposed as
\begin{equation}
g(\mathbf{\lambda}) = \sum_{s \in S}{\max_{\mathbf{r}}{\{U_{sd}(\mathbf{r}_{s}) - \mathbf{\lambda}^{T}W_{s}\mathbf{r}_{s}\}}}
\end{equation}

Thus $g(\mathbf{\lambda})$ can be computed if each node independently solves the problem
\begin{equation}
\begin{array}{cc}
\mbox{maximize} & U_{sd}(\mathbf{r}_{s}) - \mathbf{\lambda}^{T}W_{s}\mathbf{r}_{s} \\
\mathbf{r}_{s} &
\end{array}
\end{equation}
Since $U_{sd}^{(i)}$ is concave, $\mathbf{r}_{s}$ can be determined efficiently by each $s \in \mathcal{S}$.

Since (\ref{eq:flow_opt}) is upper bounded by
\begin{equation}
\label{eq:dual_opt}
\begin{array}{cc}
\mbox{minimize} & g(\mathbf{\lambda}) \\
\mbox{s.t.} & \mathbf{\lambda} \geq 0
\end{array}
\end{equation}

and since $g(\mathbf{\lambda})$ can be obtained for given $\mathbf{\lambda}$, it remains to find the optimum value of $\mathbf{\lambda}$ (i.e., $\mathbf{\lambda}^{\ast}$) satisfying (\ref{eq:dual_opt}).  By maximizing (\ref{eq:Lagrangian}) with $\lambda = \lambda^{\ast}$, the optimum solution $\mathbf{r}$ satisfying (\ref{eq:flow_opt}) can then be found.

$g(\mathbf{\lambda})$ is the pointwise maximum of a set of convex functions, and is therefore convex.  Hence $g(\mathbf{\lambda})$ can be minimized using subgradient methods, in which at each time step $t$ $g(\mathbf{\lambda})$ is updated by setting
\begin{equation}
\mathbf{\lambda}_{l}(t+1) = \mathbf{\lambda}_{l}(t) - \frac{t_{0}}{t}\left(c_{l} - \sum_{s \in S}{W_{s}\mathbf{r}_{s}}\right)
\end{equation}
%The procedure for computing the optimal source rates is described further in the Algorithm given in Fig \ref{fig:distributed_flow_alloc}.

%Since $g(\mathbf{\lambda})$ is the maximum of a family of convex function, it is convex, and can be minimized through sub gradient optimization methods.  One subgradient is given by the vector $\mathbf{v} \in \mathbf{R}^{L}$, where
%\begin{equation}
%v_{l} = c_{l} - \sum_{\pi: l \in \pi}{\mathbf{r}_{s}^{\ast}(\mathbf{\lambda})}
%\end{equation}
%At the $k$-th iteration of the algorithm, the price $\mathbf{\lambda}_{l}$ can then be updated to
%\begin{equation}
%\mathbf{\lambda}_{l}(k+1) = \mathbf{\lambda}_{l} - \alpha_{k}v_{l}(k)
%\end{equation}
%where $\alpha_{k}$ is a suitably chosen step size.

\begin{thm}
If each trust value $T_{i,j}$ converges to a steady-state value $\hat{T}_{i,j}$, then the above algorithm converges to the solution to (\ref{eq:flow_opt}) with $T_{i,j} = \hat{T}_{i,j}$.
\end{thm}

\begin{proof}
For each path $\pi$, let $\hat{T}_{\pi} = \prod_{(i,j) \in \pi}{T_{i,j}}$.  Since $\prod_{(i,j) \in \pi}{T_{i,j}}$ is a continuous function of the $T_{i,j}$'s, $\prod_{(i,j) \in \pi}{T_{i,j}}$ converges to $\hat{T}_{\pi}$.  Let $u_{\pi} = |\hat{T}_{\pi} - \tau_{T}| > 0$, then there exists $K_{\pi}$ sufficiently large that, after the $K_{\pi}$-th iteration, $|\hat{T}_{\pi} - \prod_{(i,j) \in \pi}{T_{i,j}}| < u_{\pi}$.  Hence after the $K$-th iteration, where $K = \max{K_{\pi}}$, the routing matrix is fixed, and the subgradient algorithm converges to the global optimum \cite{bertsekas1999nonlinear}.  Note that if $\hat{T}_{\pi} = \tau_{T}$, then it is impossible to bound $\hat{T}_{\pi}$ away from $\tau_{T}$, and so the trust value may oscillate infinitely around the threshold and thus fail to converge.
\end{proof}

\section{Simulation Study}
\label{sec:simulation}
 \begin{figure*}[t]
\centering
$\begin{array}{ccc}
\includegraphics[width=2.4in, height =1.6in]{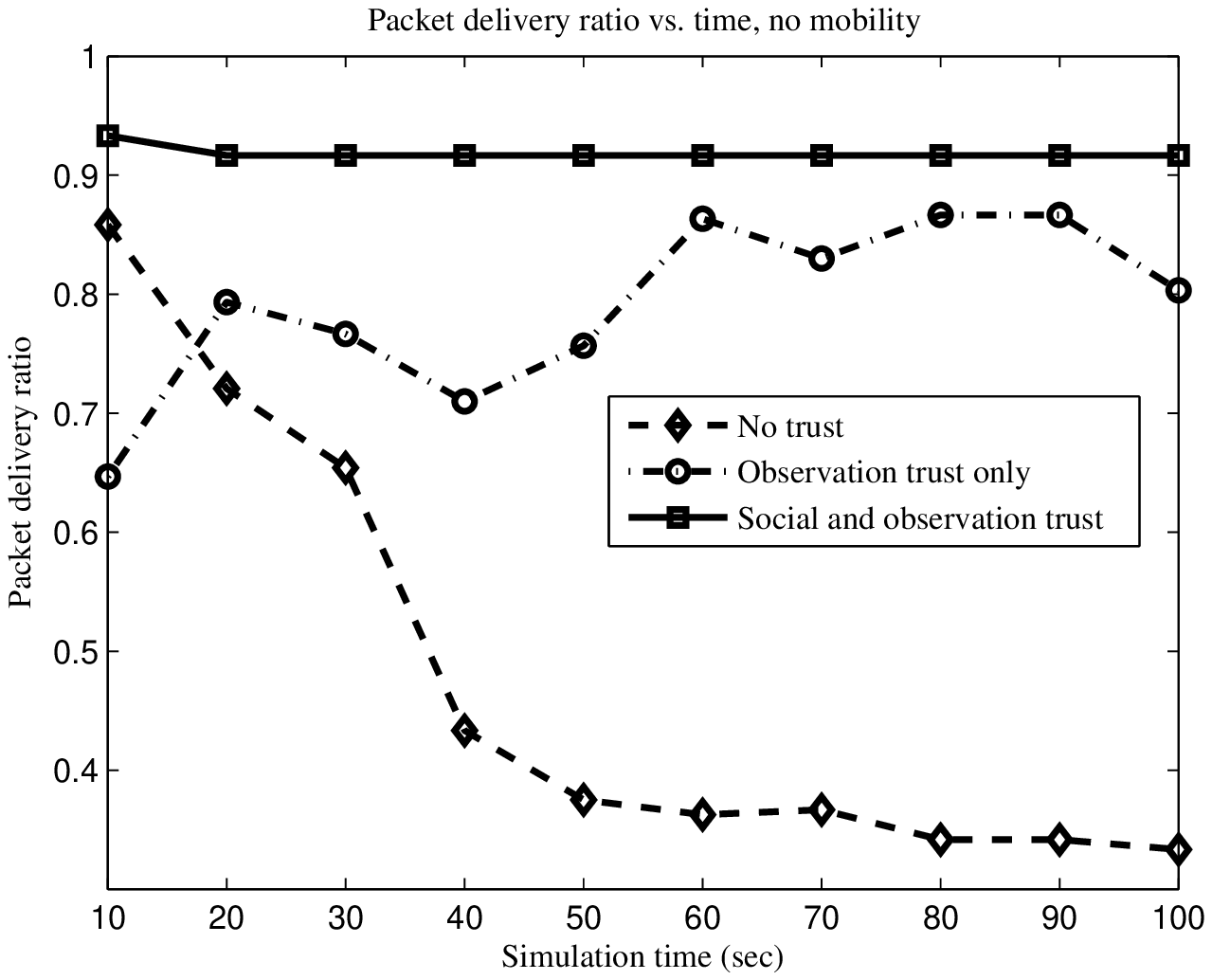} &
\includegraphics[width=2.4in, height =1.6in]{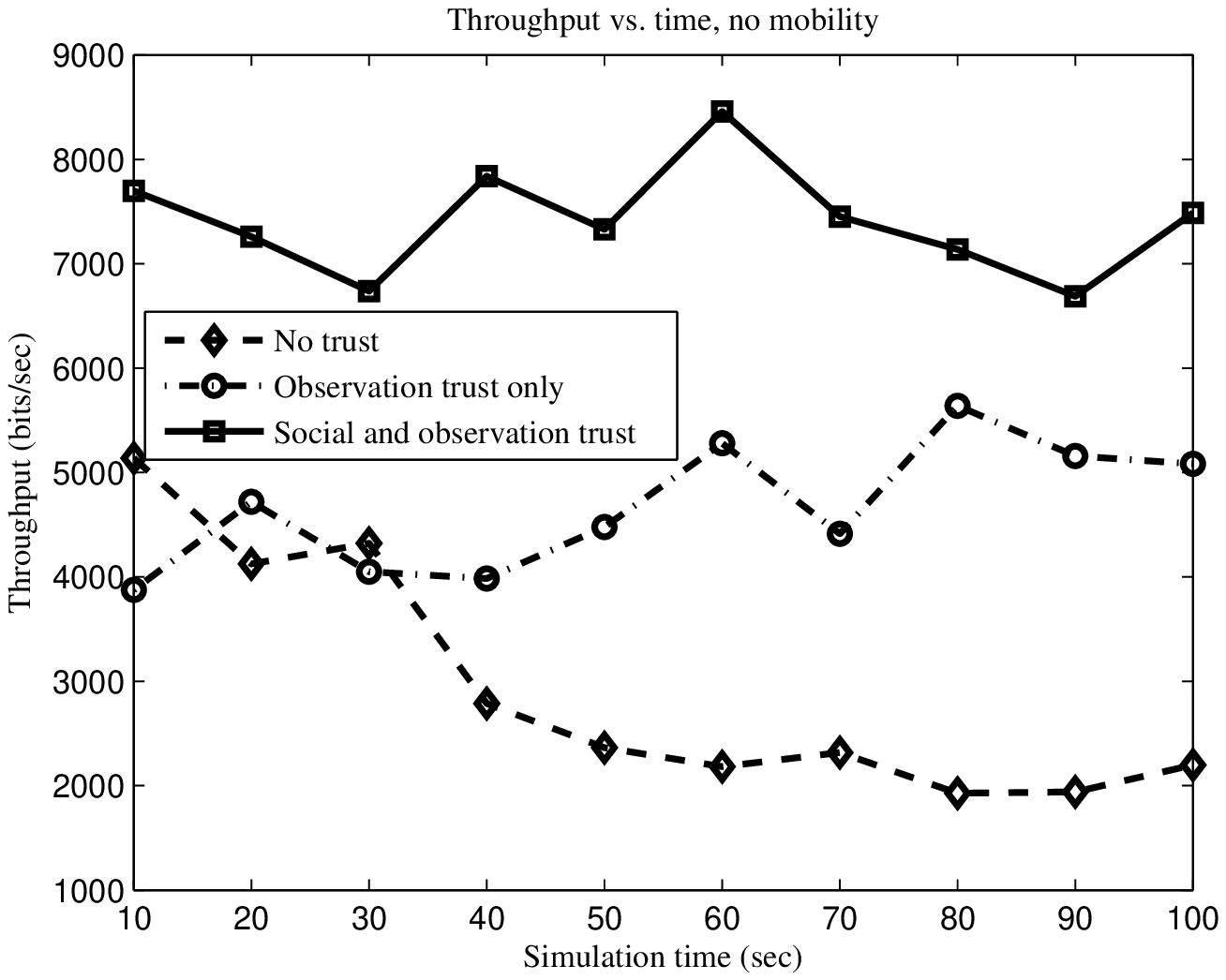} & 
\includegraphics[width=2.4in, height =1.6in]{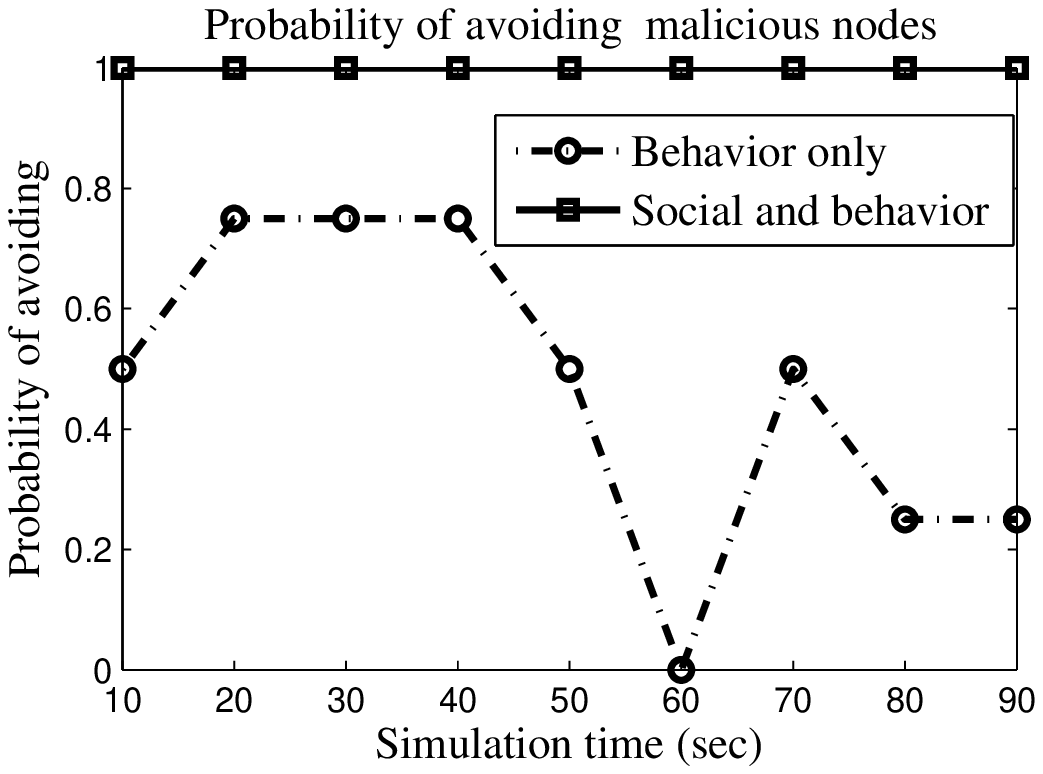} \\
(a) & (b) & (c) \\

\includegraphics[width=2.4in, height =1.6in]{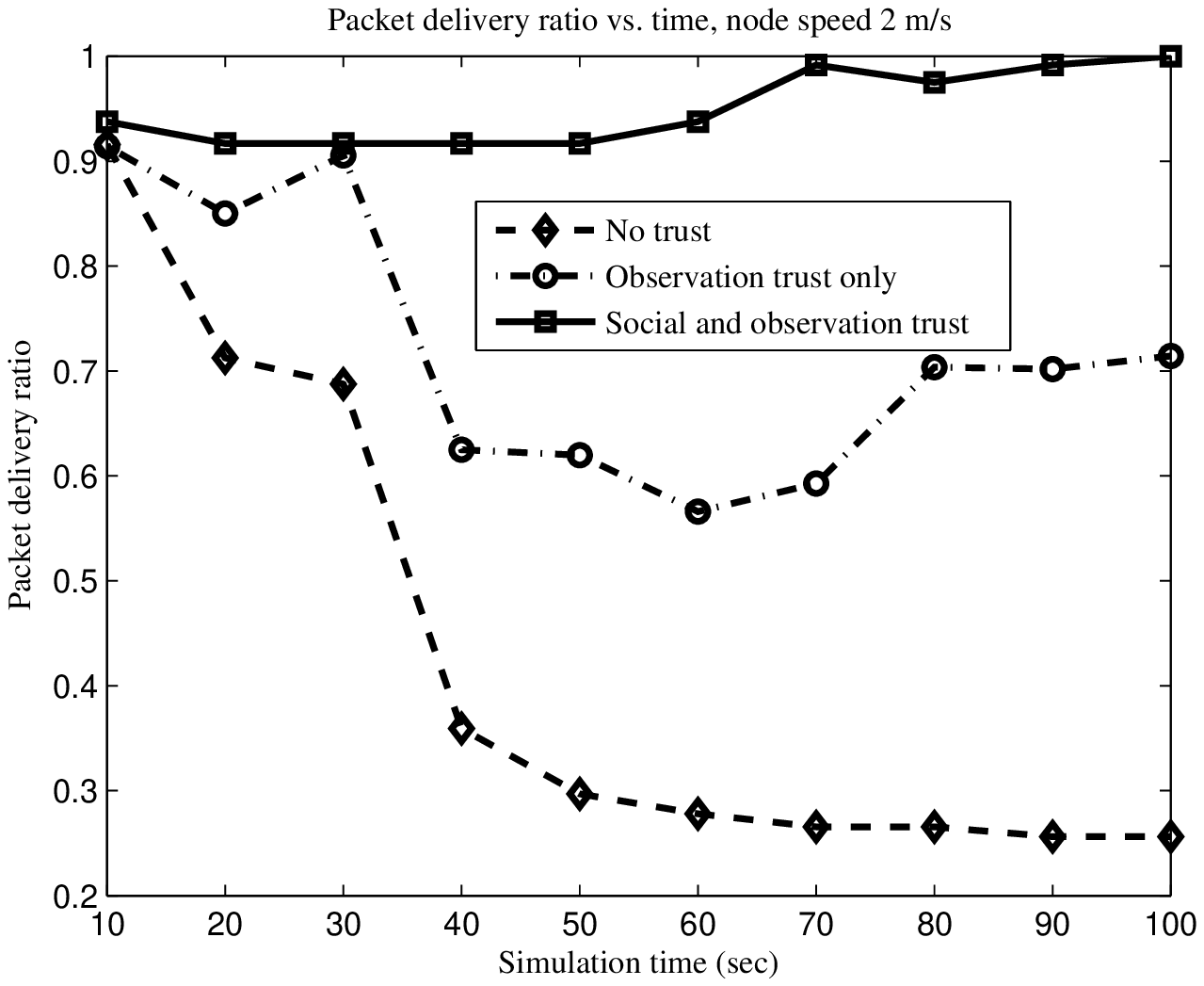} &
\includegraphics[width=2.4in, height =1.6in]{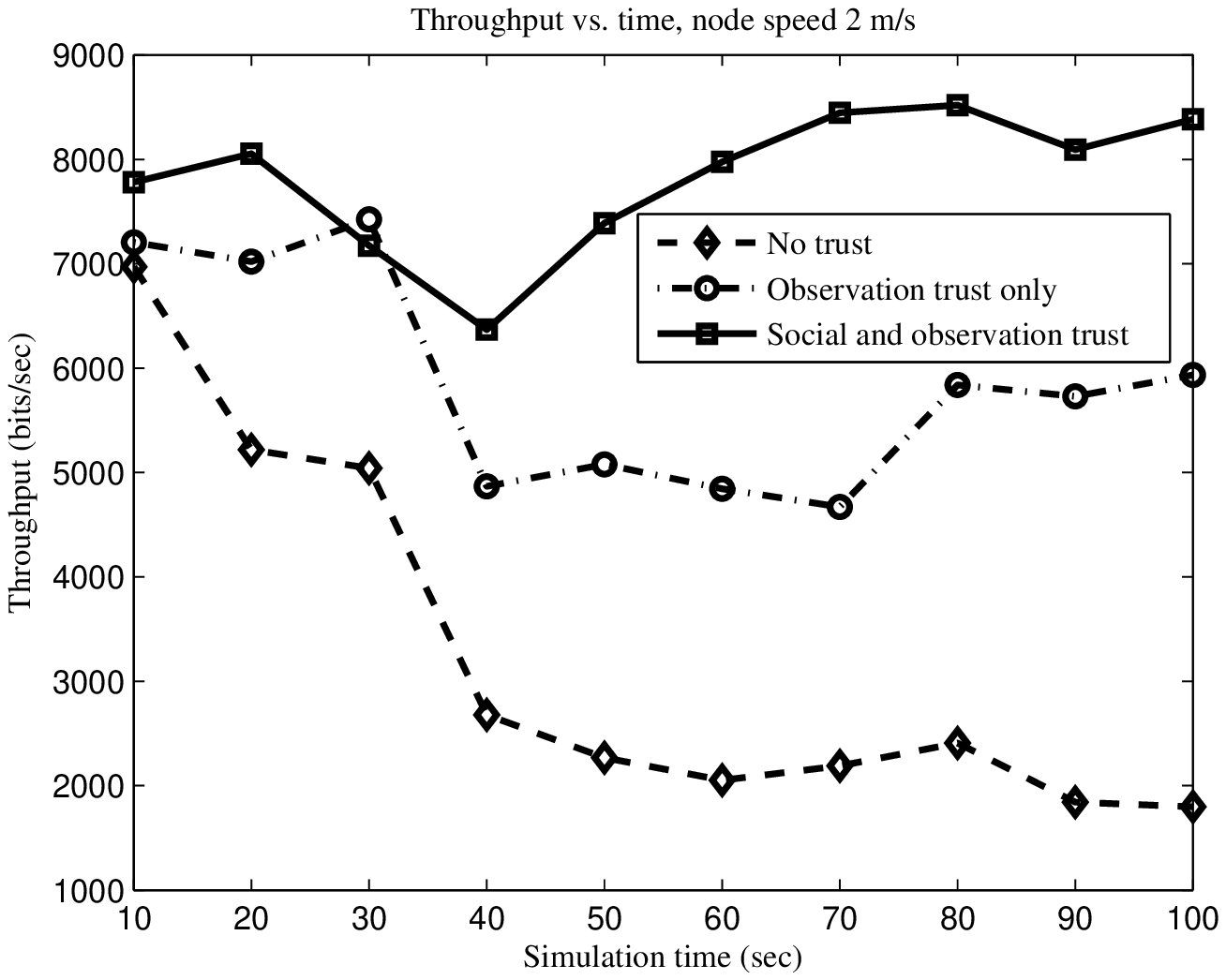} &
\includegraphics[width=2.4in, height =1.6in]{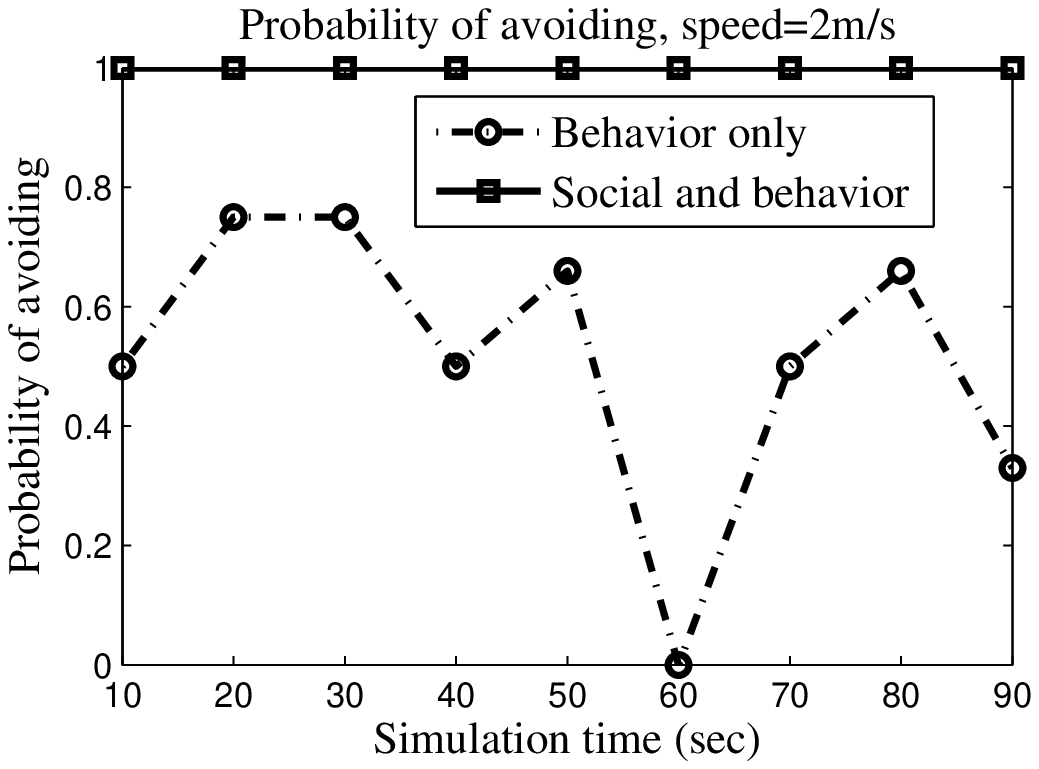}\\
 (d) & (e) & (f) \\
\end{array}$
\caption{Performance of ad hoc networks under different mobility and trust models. (a) In a static network, incorporating social trust increases the packet delivery ratio compared to behavioral trust. (b) This increase in delivery ratio results in an increase in achieved throughput.  (c) Social trust also increases the probability that packets will avoid malicious nodes. (d), (e) In a network of mobile users, there is a larger improvement in delivery ratio and throughput from social trust. (f) Probability of avoiding malicious nodes in a mobile environment.}
\label{fig:performance}
\end{figure*}

In this section, an ns-3 \cite{ns3} evaluation of our proposed flow control approach is described.

\subsection{Simulation Setup/ system model}
\label{subsec:simul_setup}

\begin{figure*}[ht]
\centering
$\begin{array}{cc}
\includegraphics[width=2.8in, height =1.6in]{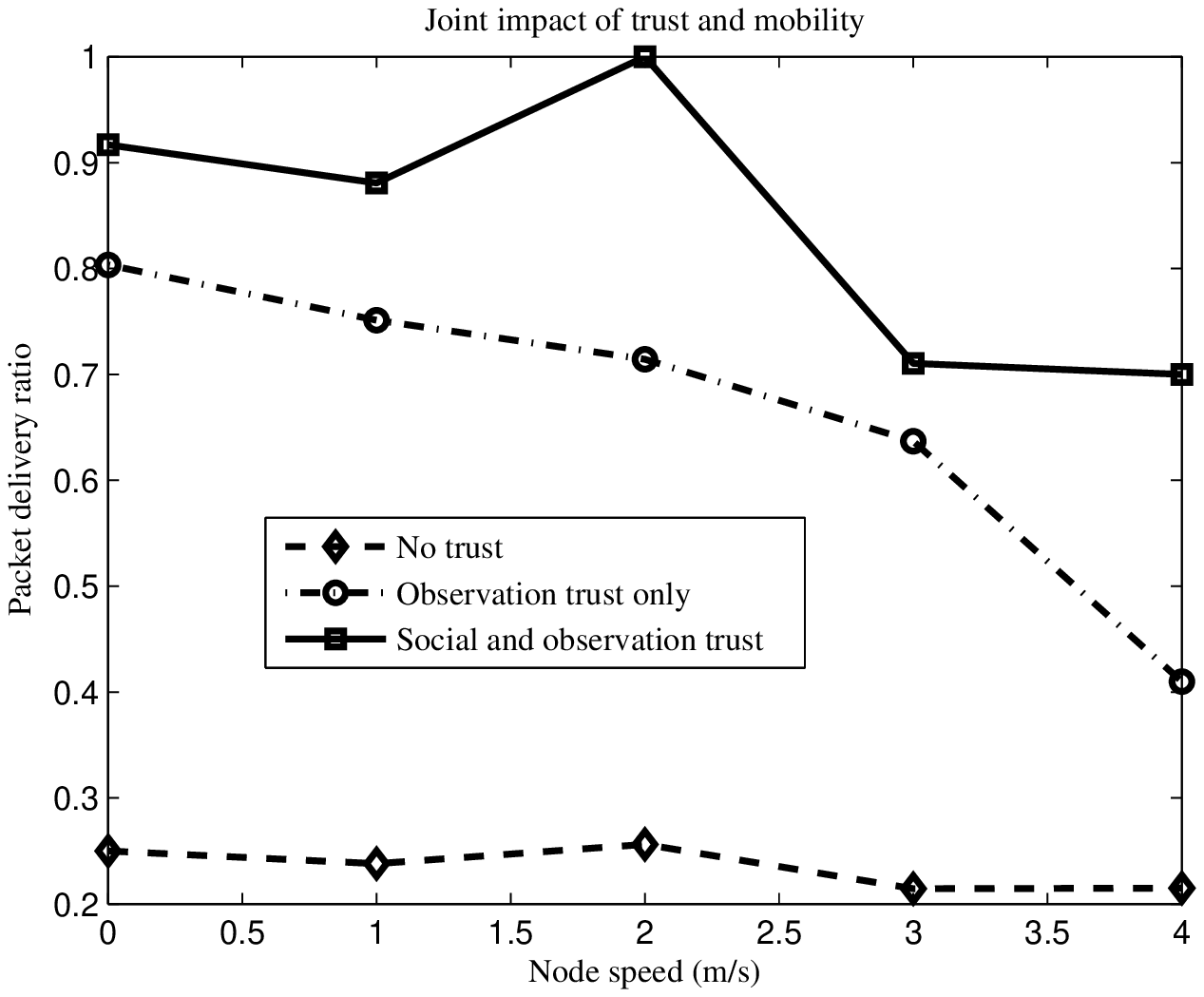} &
\includegraphics[width=2.8in, height =1.6in]{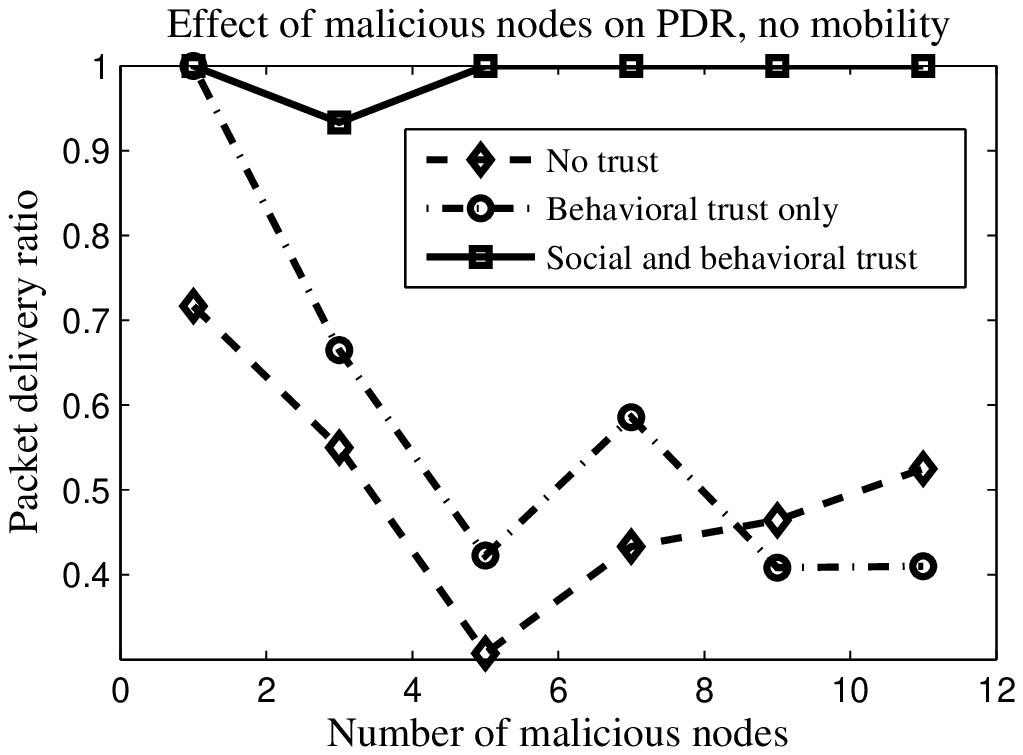} \\
(a) & (b) \\
\includegraphics[width=2.8in, height =1.6in]{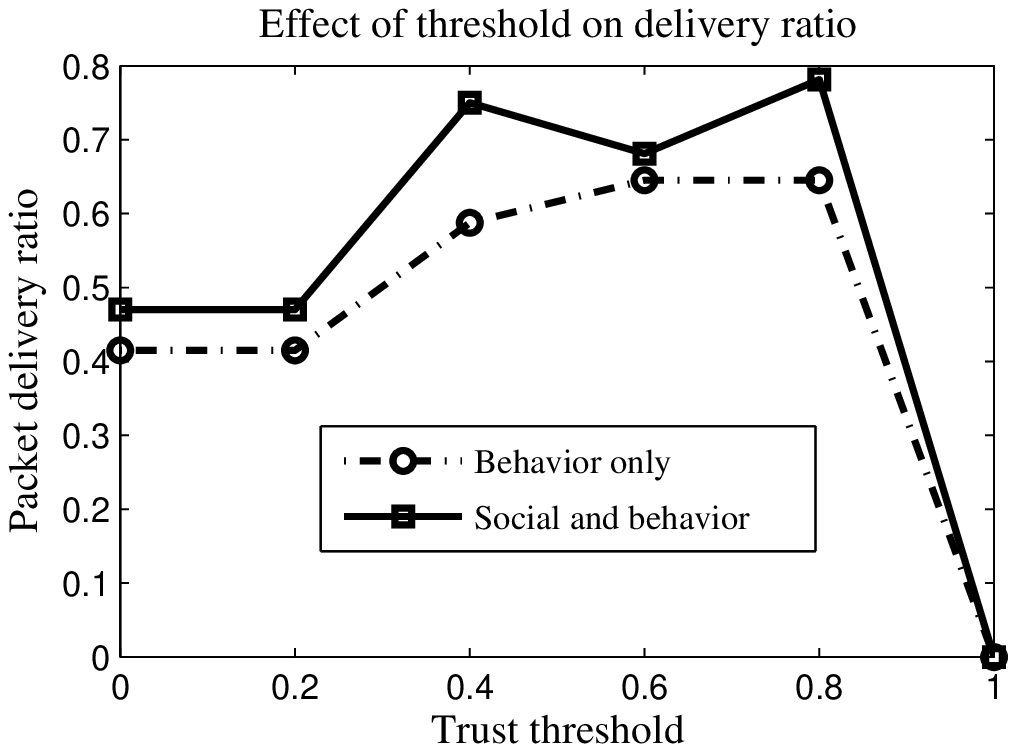} &
\includegraphics[width=2.8in, height =1.6in]{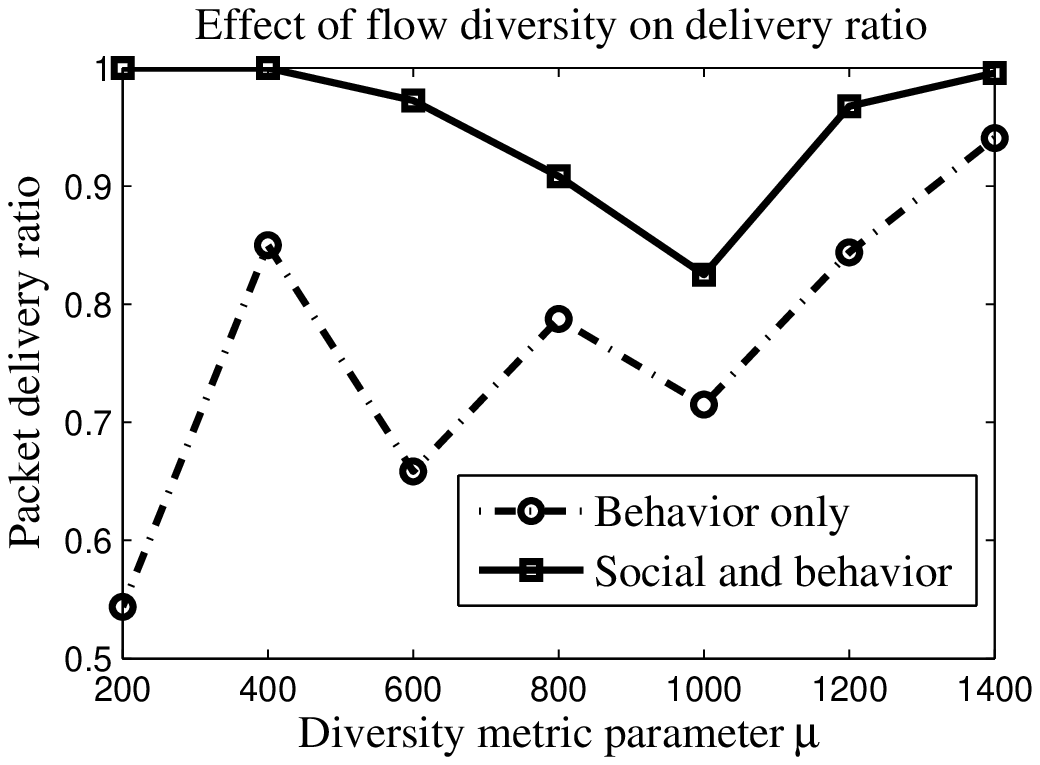} \\
(c) & (d) \\
\end{array}$
\caption{Effect of network and metric parameters on performance, as measured by packet delivery ratio. (a) Effect of mobility on packet delivery ratio using different trust metrics. (b) Performance degradation as number of malicious users increases.  (c) Effect of trust threshold on performance. (d) Effect of path diversity on delivery ratio.  Nodes move with speed 2 m/s. }%(e) Effect of identity spoofing on delivery ratio.}
\label{fig:parameters}
\end{figure*}
We make some assumptions about the nodes and malicious nodes:
\begin{enumerate}
\item All nodes could operate in promiscuous mode for neighbor monitoring and behavioral trust computation.
\item All links are bidirectional and all nodes use omni-directional transceivers.
\item Misbehaving nodes may be selfish (i.e. refuse to forward the packets of other users with probability $p$) or maliciously spoof the identity of socially trusted nodes. However, they can't obtain any certificates to verify their identities.
\item The network is a multi-hop network.
\item A malicious node will have low social trust with other users. This is justified because a user with high social trust will not behave maliciously to protect his social relationships/ reputation. Of course, he may act maliciously/ selfishly to nodes with whom he has low social trust.
\item All nodes are assumed to have access to social network data.
\end{enumerate}

A network of $70$ nodes was simulated, with initial node positions chosen uniformly at random from a rectangular region of area $1500$m x $1500$m.  Static network topologies as well as a random waypoint mobility model with node speeds varying from 1 m/s to 5 m/s were considered.  A free-space path loss propagation model was assumed, resulting in each node having a radio range of approximately $400$m.  IEEE 802.11a was used at the PHY and MAC layers.  Each data point shown represents an average over $5$ trials.

 A subset of 10 nodes was chosen uniformly at random and designated as malicious nodes.  Upon receiving a packet for forwarding, a malicious node discarded the packet without forwarding with probability $0.8$. The probability $0.8$ was chosen under the assumption that malicious nodes will sometimes forward packets correctly in order to appear cooperative.

 For each packet forwarded by a node, each neighbor was assumed to observe the forwarding behavior with probability $0.1$.  When a node $i$ sent a packet to node $j$ for forwarding and node $j$ failed to forward the packet, node $i$ used this as evidence of malicious behavior and updated node $j$'s trust value accordingly.  Conversely, correct forwarding of packets was used as evidence of good node behavior. In addition to evidence-based trust, nodes were assumed to have access to a social network, in which, based on Figure \ref{social_r1}, malicious nodes had trust value 0.6 and valid nodes had trust values chosen uniformly at random in the interval [0.85, 1]. %in which malicious nodes had trust value $0.5$ and valid nodes had trust value $0.9$.

 Four pairs of nodes, designated as $(s_{i}, d_{i})$, $i=1,\ldots, 4$, were chosen at random from the set of non-malicious nodes to act as sources ($s_{i}$) and destinations ($d_{i}$). Each source $s_{i}$ determined a set of paths to destination $d_{i}$ using a modified version of the AOMDV routing protocol \cite{marina2002ad}.  The routing protocol was modified to include trust values in route advertisement packets, allowing sources to compute the overall trust of a path.  Based on the gathered route and trust information, each source chose a flow allocation according to \eqref{eq:flow_opt}. % Measurements were gathered under three scenarios: in the first, nodes used both social and evidence-based trust values. In the second, nodes used evidence-based trust alone, while in the third no trust metrics were used.

 \subsection{Simulation Results}
 \label{subsec:results}

Figures \ref{fig:performance}a and \ref{fig:performance}b show the average packet delivery ratio and throughput, respectively, over time in the case of a static network.
Using social trust, the sources are able to identify and avoid routes containing malicious nodes, resulting in high throughput and delivery ratio.
When only behavioral trust is used, however, nodes are unable to accumulate enough evidence to avoid malicious nodes.
As a result, the performance of behavioral trust is comparable to that when no trust metrics are used.

The impact of malicious users is exacerbated when nodes are mobile.  Due to mobility, each node will have limited time to observe its neighbors, and hence cannot gather sufficient evidence to evaluate that node's reliability.  This results in a further reduction in packet delivery ratio (Fig \ref{fig:performance}d) and throughput (Fig \ref{fig:performance}e).  Since the social trust values are independent of the network topology, incorporating social trust reduces the level of evidence needed to identify malicious nodes, making the trust management scheme robust to topology changes.

The impact of mobility is further illustrated in Fig \ref{fig:parameters}a. In all three scenarios (social and evidence-based trust, evidence-based trust alone, and no trust management), there is a decrease in packet delivery ratio as node speed increases.  This is because routes become inactive at a higher rate, leading to losses when packets are sent over links that no longer exist.  Furthermore, when new routes are selected in response to topology changes, they may include malicious nodes.  Under the social trust-based method, these malicious nodes are detected faster, and hence flows are not allocated along routes containing malicious nodes.

%\begin{figure}[t]
%\centering
%\includegraphics[width=2.5in]{figures/Figure_7_malicious.eps}
%\caption{Performance degradation as number of malicious users increases.}
%\label{fig:malicious}
%\end{figure}

The effect of malicious nodes is more pronounced as the number of malicious nodes in the system increases (Fig \ref{fig:parameters}b).  When only behavioral trust is used, gathering enough evidence to avoid all malicious nodes becomes increasingly difficult.  Using social trust mitigates this effect, until the number of malicious nodes increases to the point where they cannot be avoided in routing.

%\begin{figure}[t]
%\centering
%\includegraphics[width=2.5in]{figures/Figure_8_threshold.eps}
%\caption{Effect of trust threshold on performance.}
%\label{fig:threshold}
%\end{figure}

In addition to the number of malicious nodes, network performance is affected by the trust threshold, $\tau_{T}$ (Fig \ref{fig:parameters}c). When the threshold is low, any route, even one containing untrustworthy nodes, can be used for packet delivery, resulting in malicious behavior and packet drops. As the threshold increases, the delivery ratio improves; however, if the threshold is too high, then no nodes will meet the threshold, and hence no paths can form.

%\begin{figure}[t]
%\centering
%\includegraphics[width=2.5in]{figures/Figure_9_diversity.eps}
%\caption{Effect of path diversity on delivery ratio.  Nodes move with speed 2 m/s.}
%\label{fig:diversity}
%\end{figure}

The effect of flow diversity metrics is illustrated in Fig \ref{fig:parameters}d.  For both social trust and behavioral trust, increasing flow diversity leads to a higher packet delivery ratio.  The effect is more pronounced for behavioral trust because packets are more likely to avoid undetected malicious nodes.  Moreover, by dividing the flow among multiple paths, evidence regarding a wider collection of nodes is gathered, allowing more rapid detection of malicious nodes.

\begin{figure}[t]
\centering
\includegraphics[width=2.8in, height =1.5in]{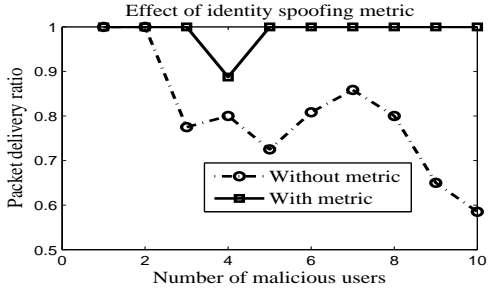}
\caption{Effect of identity spoofing on delivery ratio.}
\label{fig:vulnerability}
\end{figure}

 % To simulate the effect of identity spoofing on system performance, it was assumed that a set of malicious users spoof the identity of  trusted users from the online social network.  It was further assumed that malicious users were not able to obtain any certificates to verify their identities.  Without the use of the identity spoofing metric, this attack may remain undetected, resulting in a drop in performance
In earlier simulations, ISM was incorporated. Figure \ref{fig:vulnerability} shows the comparison when ISM is disabled.  By incorporating the metric, the probability of routing packets through a node with a spoofed identity is greatly reduced (except a random case for 4 malicious users).

\section{Related Work}
\label{sec:related}
There is active ongoing research in establishing trust in the wireless networks, as part of the broader topic of security in ad hoc networks~\cite{hubaux2001quest}.% In wired networks, trust systems such as PGP~\cite{garfinkel1995pgp} have been successfully deployed.  Trust-based security mechanisms using small world concepts to optimize formation and propagation of trust amongst network entities is proposed in \cite{small_world}.
 An authentication and secure channel establishment protocol for trust propagation for multihop wireless home networks is proposed in \cite{trust_propagation}. A distributed approach that establishes reputation-based trust among sensor nodes in order to identify malfunctioning and malicious sensor nodes and minimize their impact on applications is proposed in \cite{statistical_trust}. A formalism to design trust as a computational component is proposed in \cite{trust_computational}.
 % An experiment has been designed in \cite{trust_reciprocity} to study trust and reciprocity in an investment setting.
 % A trust establishment scheme which aims to improve the reliability of packet forwarding over multi-hop routes in the presence of potentially malicious nodes is described in \cite{trust_secure_rout_7}. A framework for trust propagation schemes is developed in \cite{trust_Distrust_guha}.

 A detailed survey on trust computations in MANETs is carried out in \cite{kannan_trust} and on trust management in \cite{jinhee_trust}. These studies, however, do not consider the use of social  trust in an ad hoc network scenario.

 %Moreover, while some existing studies have considered the impact of trust \cite{theodorakopoulos2006trust} on route selection, trust metrics have not been incorporated into network flow allocation.

Flow allocation and rate control have been studied using optimization methods \cite{kelly1998rate}.  Dual decomposition methods, in which constraints for each source are decoupled in order to develop distributed algorithms, are one widely-used optimization method \cite{chiang2007layering}.  These methods, however, do not incorporate the trustworthiness of nodes or the possibility of node misbehavior, and hence cannot be guaranteed in the presence of malicious users.

Routing in the presence of malicious nodes has been considered in a variety of existing works~\cite{marti2000mitigating,karlof2003secure}.  %Secure routing aims to protect the network from a variety of network layer attacks, including broadcasting a large number of Hello or route request messages in order to create congestion, or tampering with routing packets in order to create routing loops.
 Malicious nodes can also assume multiple node identities (the Sybil attack~\cite{newsome2004sybil,yu2006sybilguard}).  Rather than preventing an adversary from disrupting the routing protocol, however, our focus is on identifying potential malicious nodes and allocating flows that avoid them.%We observe, however, that even if the routing protocol is robust to attack, an adversary acting as an intermediate relay can still modify the contents of packets

There have been some efforts to enable social networking in ad hoc networks~\cite{li09mobisn,sarigl09} but the effect of social relationships on network performance has not been studied so far. MobiSN~\cite{li09mobisn} implements all of the core features of a mobile ad hoc social networking including profile generating, friend matchmaking, routing control, ad hoc multi-hop text messaging and file sharing. It is implemented in Java 2 Micro Edition (J2ME) for Bluetooth-enabled mobile phones. Sarigol et al~\cite{sarigl09}
 present Adsocial, designed to run on resource-constrained mobile devices and share data using a simple and efficient data piggybacking mechanism. %\cite{sarigol2010tuple} extends the idea and proposes a distributed tuple space for social networking on ad hoc networks. %\cite{gurecki09} presents adhocClient, an Android application that discovers other clients using a UDP heartbeat message. When one or more clients are discovered, the various clients can exchange status updates, messages, and GPS coordinates from each other.

Our work  differs substantially compared to all the related work mentioned above. Apart from uniquely combining trust derived from social media with the behavioral trust, we also incorporate the composed metric into network flow allocation. In our knowledge this is an unique attempt in this direction.

\section{Conclusion}
\label{sec:conclusion}
In this paper, we proposed a flow allocation problem based on trust metrics. We introduced a framework for integrating social trust into mobile ad hoc networks. As a first step, we demonstrated how to combine the social and observation trust into a single trust metric and update this value over time. We have provided social trust calculation strategy based on real time traces.  A flow allocation optimization framework was then developed for allocating network flows based on trust metrics, incorporating trust, throughput, and capacity constraints.  Distributed algorithms for obtaining the optimal throughput were provided.

 We achieve an improvement in terms of throughput with our proposed approach (social plus behavioral trust) than with behavioral trust alone. Our scheme is also capable of avoiding malicious nodes during routing and providing higher packet delivery ratio as the network operation time increases.  This implies that both behavioral and social trusts are important and with together they provide a significant performance improvement. Especially social trust can play a crucial role in providing a significantly higher overall end-to-end performance in mobile networks.

%\textbf{Additional plots I will generate: }
%\begin{enumerate}
%\item Effect of increasing the network size and number of compromised nodes
%\item Effect of varying the parameters $\sigma$ and $\tau$
%\end{enumerate}
%For fixed $\lambda$, each source can choose the set of rates $\{r_{s\pi}\}$ maximizing $U_{s}(\mathbf{r}_{s}) - \lambda^{T}\left(\sum_{s}{W_{s}r_{s}}\right)$ through convex optimization.  In order to arrive at the optimal $\mathbf{\lambda}$, each link $l$ iteratively updates the value of $\lambda_{l}$ based on the rule
%\begin{equation}
%\lambda_{l}(t+1) = \lambda_{l}(t) - \frac{t_{0}}{t}\left(c_{l} - \sum_{s \in %S}{W_{s}\mathbf{r}_{s}}\right)
%\end{equation}

%It can be shown that this update leads to the global optimum value of $\lambda$, denoted $\lambda^{\ast}$.  Each node, by maximizing $U_{s}(\mathbf{r}_{s}) - \lambda^{T}\left(\sum_{s}{W_{s}r_{s}}\right)$, can then arrive at the value of $\mathbf{r}_{s}$ solving (\ref{eq:flow_opt}]).

\bibliographystyle{IEEEtran}
\bibliography{social_trust}

\end{document}